\documentclass[journal]{IEEEtran}

\usepackage[T1]{fontenc}
\usepackage[english]{babel}
\usepackage{color}
\usepackage{xcolor}

\usepackage{setspace}
\usepackage{longtable}
\usepackage{array}
\usepackage{cite}

\ifCLASSINFOpdf
\usepackage[pdftex]{graphicx}
\usepackage{float}   
\else
\fi

\usepackage{amsmath,amsfonts,amssymb,mathtools,amsthm,mathrsfs}
\usepackage{enumitem}
\usepackage[dvipsnames]{xcolor}
\usepackage[normalem]{ulem}
\usepackage{makecell}
\usepackage{subfigure}
\usepackage{algpseudocode}

\theoremstyle{definition}

\newtheorem{theorem}{Theorem}

\newtheorem{lemma}[theorem]{Lemma}

\usepackage{nomencl}
\makenomenclature
\usepackage{array,multirow,graphicx}
\usepackage{soul}
\usepackage{layout,array}
\usepackage{ragged2e}
\usepackage{algorithm}
\usepackage{algpseudocode}
\graphicspath{{figures/}}
\usepackage{tikz}
\usetikzlibrary{shapes.geometric, arrows}
\usepackage{mathtools}

\usepackage{etoolbox}

\renewcommand\nomgroup[1]{%
  \item[\bfseries
  \ifstrequal{#1}{A}{Indices and Sets}{%
  \ifstrequal{#1}{B}{Parameters and Variables}{%
  \ifstrequal{#1}{C}{Functions}{}}}%
]}

\newcommand{\probP}{\text{I\kern-0.15em P}}

\begin{document}
\title{
Computationally Transparent Power Flexibility Aggregation
via Directional Support Evaluation}
\author{Mingyu Zhai,~\IEEEmembership{Student Member,~IEEE} Liang~Du,~\IEEEmembership{Senior Member,~IEEE}
Yue~Chen,~\IEEEmembership{Senior Member,~IEEE}
\thanks{M. Zhai and L. Du are with the Department of Electrical Engineering, Villanova University, Villanova, PA 19085 USA. E-mail: \{mzhai,~ ldu\}@villanova.edu.}
\thanks{Y. Chen is with the Department of Mechanical and Automation Engineering, The Chinese University of Hong Kong, Hong Kong, China. E-mail: yuechen@mae.cuhk.edu.hk.}
}

\maketitle

\begin{abstract}
The growing penetration of distributed energy resources (DERs) creates
new opportunities for improving resource adequacy, operational
flexibility, and power-system reliability. However, efficiently
characterizing the aggregate flexibility of heterogeneous DER portfolios
remains challenging because of high dimensionality, diverse operating
dynamics, and spatio-temporal coupling constraints. This paper proposes
a computationally transparent framework that formulates aggregate
flexibility characterization as a sequence of linear programming
(LP)-based directional support evaluations. The framework accommodates
heterogeneous DER technologies and configurable device-, feeder-, and
portfolio-level constraints. The convex hull of the resulting boundary
trajectories provides a feasible inner approximation of the coupled
aggregate power flexibility region. Theoretical analysis establishes
feasibility and non-overestimation guarantees and characterizes the
computational dependence on portfolio size, scheduling horizon, and
direction budget. In the small-scale case, the proposed method produces
tighter projected regions than box, zonotope, and Chebyshev inner-ball
benchmarks and matches a dense directional reference within numerical
tolerance using 16--64 construction directions. Large-scale simulations
with up to 2000 DERs demonstrate approximately linear empirical runtime
scaling with the portfolio size and direction budget over the tested
ranges. The results establish a scalable and theoretically grounded
approach for aggregate DER flexibility characterization.
\end{abstract}




\section{Introduction}

The proliferation of demand response resources (DRRs) and ever-growing penetration of distributed energy resources (DERs), such as renewable generations, alternative energy resources, electrified transportation, and multi-scale battery energy storage systems (BESS), conveys a profound paradigm shift in the control of power grids. Specifically, uncertainties caused by DER penetration and climate change intertwine and amplify challenges posed on system reliability. 

To address the imminent technical challenges, 
the concept of power flexibility has been rejuvenated as the ability of distribution feeders to actively control their power injections (at the substation interface) and/or consumption patterns in response
to external control or market signals, which is also known as power flexibility aggregation~\cite{10738509}. 
However, accurately quantifying and aggregating the flexibility potential of DERs remains a significant challenge~\cite{hamidi2025performance}.
Specifically, integrating massive, heterogeneous DERs, BESS, and DRRs through incentives and aggregation mechanisms against operational constraints across devices as well as stochastic end-user behaviors is often computationally intractable~\cite{7864461} or impractical~\cite{trangbaek2012exact}
if an exact, feasible power flexibility region (PFR) is desired by aggregation. Moreover, aggregating individual flexibility profiles into a system-level representation remains a non-trivial task due to the high dimensionality and the combinatorial nature of feasible operating regions. Consequently, exact flexibility aggregation is computationally intractable in general and typically exhibits NP-hard complexity~\cite{mukhi2025robust}.

To address the aforementioned computational challenges, many methods have been proposed  in the literature, which can be grouped by four main categories: inner approximation methods, outer approximation methods, directional/projection-based, and data-driven learning-based methods~\cite{wen2022aggregate,10738509,he2026analyzing}. Each category offers a different trade-off between computational efficiency, modeling fidelity, and operational feasibility.
 
\textit{\textbf{Inner approximation methods}} construct conservative subsets of the actual aggregated feasible region, such as polytopes~\cite{chen2019aggregate}, zonotopes~\cite{10576684}, and balls~\cite{11225608}. These methods guarantee that all points within the approximated region are operationally feasible, thereby preserving system reliability and physical consistency. However, the main drawback is that these techniques often underestimate the actual flexibility potential and lead to overly cautious control decisions or reduced market participation revenue. For example, the work in~\cite{yuan2025aggregation} introduces an iterative scheme based on high-dimensional polyhedral projection to improve the efficiency of inner approximation method while maintaining its tightness. Although this approach improves performance relative to naive convex bounding techniques, it still sacrifices some flexibility volume for the sake of guaranteed feasibility. Therefore, some other efforts explore improved strategies to ameliorate approximation methods achieving more effective and efficient approximation. A closed-form exact aggregation model is utilized in~ \cite{wen2023improved} to systematically design a tighter and computationally efficient inner approximation. Meanwhile, a Wasserstein distributionally robust joint chance-constraint framework is designed in~\cite{zhou2024aggregated} to characterize aggregated feasible active power regions with less conservative outcomes and reduced computational complexity. Similarly, \cite{altaha2024efficient} proposes to find the largest safe inner polytope inside a Minkowski sum and derive conservative linear programming (LP) approximations to improve its tightness and computational performance. To summarize, existing inner approximation methods mainly focus on how to improve approximation tightness while maintaining feasibility guarantees and reducing computational complexity, enabling flexibility aggregation for large DER populations.

\begin{table*}[t] 
\label{review}
\caption{Comparison with the Existing Literature}
\centering
\footnotesize
\begin{tabular}{lccccccc}
\hline Reference & \makecell{Approximation\\Method}& 
DERs & 
\makecell{Network\\Constraints} &
\makecell{Time-Coupled\\(Heterogeneity)}&
\makecell{Computational \\Time Comparison} & 
\makecell{Computational Complexity\\Theoretical Analysis} & \\
\hline~\cite{wen2023improved}, \cite{altaha2024efficient}, \cite{zhang2023fast} & Inner & EVs & $\times$ & $\checkmark$ & $\checkmark$ & $\times$ \\
\hline~\cite{yi2021aggregate} & Inner & BESS,TCLs & $\times$ & $\times$ & $\checkmark$ & $\times$ \\
\hline~\cite{wen2022aggregate}, \cite{chen2019aggregate} & Inner & PV,BESS,DCL,HVAC & $\checkmark$ & $\times$ & $\checkmark$ & $\times$ \\
\hline~\cite{li2024modeling}, \cite{li2026projection} & Projection & DG,BESS,EVs,PV,HVAC & $\times$ & $\checkmark$ & $\checkmark$ & $\times$ \\
\hline~\cite{barot2017concise}, \cite{kundu2018approximating}  & Outer & BESS,TCLs,PV & $\times$ & $\times$ & $\times$ & $\times$ \\
\hline~\cite{li2023distribution},\cite{wang2025bayesian} & Data-driven & Loads,PV,ESS,DCL & $\checkmark$ & $\checkmark$ & $\times$ & $\times$ \\
\hline~This work & Projection & PV,LIB,EV,ZAB,SFB & $\checkmark$ & $\checkmark$ & $\checkmark$ & $\checkmark$ \\
\hline
\end{tabular}
\end{table*}

\textit{\textbf{Outer approximation methods}}, by contrast, construct supersets that enclose the actual aggregated flexibility region. These methods are designed to offer aggressive flexibility estimation, enabling more proactive operational strategies and increased utilization of DERs. However, the major drawback lies in their tendency to include infeasible solutions, i.e.,  power flexibility points which cannot be realized due to physical or temporal constraints. For instance,~\cite{barot2017aggregate} proposes an approach based on computing outer bounds of the Minkowski sum of individual feasible regions. This technique achieves improved computational tractability and expands flexibility representation at the cost of including points that violate local constraints, which can lead to dispatch errors, disaggregation infeasibility or compromise loads' traceability. 
As a comparison, 
while inner approximations favor safety and robustness, outer approximations promote higher (potential) gains with higher risk.

\textit{\textbf{Directional/projection-based methods}} characterize aggregated power flexibility regions by repetitively solving optimization problems along sampled directions, evaluating feasible set, and simultaneously maximizing or minimizing the aggregate power profile. Consequently, a set of boundary points or low-dimensional projections of the true regions are constructed. These methods are well-suited for scalability and high-dimensional settings since each projection direction can be solved independently, thereby naturally enabling parallel computation~\cite{liu2025coupling}. In~\cite{li2026projection}, a projection-based maximal inner approximation method is proposed to approximate the aggregated flexibility of large-scale heterogeneous DERs with diverse operational states and time horizons. 
Compared with other methods, projection-based methods can achieve task-oriented direction/preference selection rather than purely geometric conservative bounding. For instance, the Minkowski sum of individual flexibility sets can be reformulated into an equivalent polytope projection problem using a matrix-transformation technique \cite{liu2025preference}, which embeds reserve-dispatch preferences into the aggregation process, mitigates excessive conservativeness, and improves its practical usefulness.
However, too many samples may be required to achieve a tight approximation of the feasible boundary, especially when the aggregated region is highly anisotropic or contains sharp facets, which leads to high computational burden.

\textit{\textbf{Data-driven learning-based
methods}} leverage operational data and machine learning models to approximate the aggregated power flexibility region without explicitly modeling efforts, which reduce online computational burden and naturally capture complex nonlinearities and uncertainties. However, a key limitation is that these approximations typically lack explicit feasibility guarantees and are highly sensitive to data coverage, model generalization, and distributional shifts. In~\cite{li2023distribution}, a network-informed data-driven framework is proposed to generate high-quality training data for a classifier to estimate the feasible power output region with high accuracy. This work also highlights learning-based surrogates as a scalable approach to high-dimensional characterization when explicit polytope construction is computationally prohibitive.
In \cite{mukhi2026exact}, a data-driven inverse optimization framework is developed to model the spatial correlations between renewable generation and demand and to capture correlated uncertainty for flexibility characterization. Finally, our previous effort~\cite{wang2025bayesian} proposes a network-aware, near-optimal, and sample-efficient data-driven method, in which inner-box and outer-box approximated regions are first constructed for bus- and feeder-level PFRs, respectively. A Bayesian active learning-based soft calibration strategy is proposed to iteratively and interactively learn a decision boundary of the binary classifier as an improved feeder-level PFR while providing trustworthy measurement to system operators.


As summarized 
in Table~\ref{review}, the existing literature mostly focuses on improving approximation performance while reducing computational costs, which often overlooks theoretical complexity analysis and lacks transparent insights into the practicality and scalability aspects. Existing studies typically report numerical runtime comparisons, but these are case-dependent and provide limited insight into algorithmic efficiency. To our best knowledge, a transparent, general, and theoretically grounded complexity analysis—one that explicitly characterizes asymptotic scaling with respect to system size, time horizon, and constraint coupling—remains a major gap. This work aims at establishing a tunable complexity models to enable fair cross-method comparisons and support improvements in computational efficiency and real-time applicability.

Furthermore, most existing studies characterize feasible power flexibility regions by focusing on conventional DER types. Consequently, estimation methods should be more generalizable and DER-agnostic, so that they can accommodate heterogeneous portfolios without requiring substantial reformulation. For instance, driven by the rapid growth of data centers and AI-training infrastructure, emerging energy storage technologies (e.g., zinc–air batteries (ZAB) and solid-fuel batteries (SFB)) are gaining increasing attention~\cite{rahman2026energy}, yet their operational characteristics differ fundamentally from conventional lithium-ion batteries. Therefore, developing estimation methods that can seamlessly incorporate heterogeneous new DER devices while maintaining computational efficiency constitutes an impactful research direction.

To address these gaps, this paper proposes a computationally
transparent framework for characterizing the aggregate flexibility of
heterogeneous DER portfolios through adaptive directional support
evaluations. As illustrated in Fig.~\ref{framework}, the framework
integrates device-level operating constraints and system-level coupling
limits while providing a tunable trade-off between approximation
accuracy and computational effort. A unified complexity analysis is
further developed to characterize the principal factors governing
runtime and scalability.

The main contributions of this work are fourfold:
\begin{enumerate}
    
\item \textbf{A unified LP-based directional operator for scalable
flexibility aggregation:}
The aggregate flexibility problem is formulated as a sequence of
LP-based support evaluations under configurable operational and coupling
constraints. The resulting feasible boundary trajectories construct an
inner approximation without explicitly enumerating the vertices or
facets of the high-dimensional aggregate flexibility region.

\item \textbf{DER-agnostic modeling incorporating emerging long-duration
storage technologies:}
A unified modeling layer is developed for heterogeneous DERs, including
EVs, lithium-ion batteries, PV systems, zinc--air batteries, and
solid-fuel batteries. In particular, nonconventional refueling and
fuel-regeneration dynamics are incorporated without changing the
structure of the aggregation framework.

\item \textbf{Provable feasibility and non-overestimation of the
aggregated PFR:}
Each boundary trajectory generated by the proposed support formulation
belongs to the coupled aggregate flexibility region. Consequently, the
convex hull of the sampled trajectories provides a feasible inner
approximation and excludes infeasible flexibility claims. This guarantee
distinguishes the proposed framework from outer approximations and
learning-based surrogates without explicit feasibility guarantees.

\item \textbf{Complexity characterization and large-scale validation:}
A theoretically grounded analysis characterizes the dependence of the
computational cost on the number of DERs, scheduling horizon, model
constraints, and projection directions. Numerical studies with
portfolios of up to 2000 DERs demonstrate approximately linear empirical
scaling with the portfolio size and direction budget over the tested
ranges.
    
\end{enumerate}

The remainder of this paper is organized as follows. Section~II presents
the device-level models, aggregate flexibility formulation, proposed
directional approximation method, and its feasibility guarantee.
Section~III analyzes the computational complexity of the proposed
method. Section~IV presents the benchmark comparisons, direction sensitivity analysis, and large-scale scalability results. Section~V
concludes the paper.

\section{Aggregate Flexibility Formulation}
\label{method}
As illustrated in Fig.~\ref{framework}, the proposed framework begins with a flexibility request and a joint convex model containing device-level operating constraints and aggregator-level coupling limits. Construction directions are generated in batches by sampling Gaussian vectors and normalizing them to unit Euclidean length. For each direction, a support LP returns a feasible aggregate net-power trajectory. After each batch, the maximum normalized support deficit is evaluated using a separate set of validation directions, whose true support values are computed once and cached. Direction sampling terminates when the validation deficit falls below the prescribed tolerance, its improvement remains below the saturation threshold for a specified number of consecutive batches, or the maximum direction budget is reached. The convex hull of the retained trajectories forms a feasible inner representation of the aggregate power region under the continuous convex model.

\begin{figure}[t]
    \centering
    \includegraphics[width=0.48\textwidth]{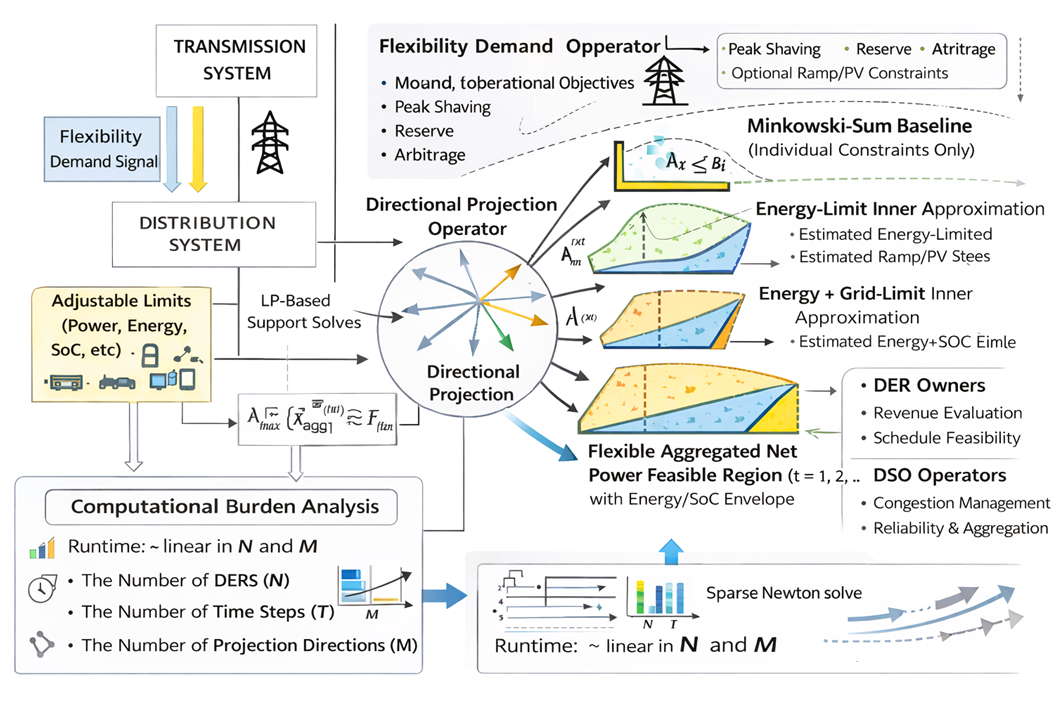} 
    \vspace{-1em}
    \caption{Overview of the proposed framework.} 
    \label{framework}
\end{figure}

\subsection{Device-Level Flexibility Models}\label{AA}
\textit{1) EVs:} 
EVs operate in vehicle-to-grid (V2G) mode, which permits bidirectional power exchange, or vehicle-one-grid (V1G) mode, which permits controlled charging only. 
Their flexibility is described by operational constraints~\cite{bao2024novel,11104140}:
\begin{equation}
\begin{aligned}
\underline{p}_{i,\mathrm{EV}}^{c,t}
&\le p_{i,\mathrm{EV}}^{c,t}
\le \overline{p}_{i,\mathrm{EV}}^{c,t},&
\underline{p}_{i,\mathrm{EV}}^{d,t}
&\le p_{i,\mathrm{EV}}^{d,t}
\le \overline{p}_{i,\mathrm{EV}}^{d,t},\\
\underline{E}_{i,\mathrm{EV}}^{c,t}
&\le E_{i,\mathrm{EV}}^{c,t}
\le \overline{E}_{i,\mathrm{EV}}^{c,t},&
\underline{E}_{i,\mathrm{EV}}^{d,t}
&\le E_{i,\mathrm{EV}}^{d,t}
\le \overline{E}_{i,\mathrm{EV}}^{d,t},\\
E_{i,\mathrm{EV}}^{c,t}
&=\sum\nolimits_{\tau=1}^t p_{i, \mathrm{EV}}^{c}(\tau) \Delta t,&
E_{i,\mathrm{EV}}^{d,t}
&=\sum\nolimits_{\tau=1}^t p_{i, \mathrm{EV}}^{d}(\tau) \Delta t,
\end{aligned}
\label{eq:ev_model}
\end{equation}
where \(p_{i,\mathrm{EV}}^{c,t}\) and
\(p_{i,\mathrm{EV}}^{d,t}\) are the charging and discharging powers,
while \(E_{i,\mathrm{EV}}^{c,t}\) and
\(E_{i,\mathrm{EV}}^{d,t}\) are the corresponding cumulative
energies. V1G operation is obtained by imposing
\(p_{i,\mathrm{EV}}^{d,t}=0\).

\textit{2) Battery and Alternative Energy Storage Systems:} This study considers distributed battery storage systems, each constrained by charging and discharging limits, as well as energy capacity. Moreover, BESS should also satisfy energy balance constraints over time.
\begin{equation}
\begin{gathered}
\underline{p}_{j,\mathrm{BS}}^{c,t}
\le p_{j,\mathrm{BS}}^{c,t}
\le \overline{p}_{j,\mathrm{BS}}^{c,t},
\quad
\underline{p}_{j,\mathrm{BS}}^{d,t}
\le p_{j,\mathrm{BS}}^{d,t}
\le \overline{p}_{j,\mathrm{BS}}^{d,t},\\
\underline{\mathrm{SoC}}_{j,\mathrm{BS}}^{t}
\le \mathrm{SoC}_{j,\mathrm{BS}}^{t}
\le \overline{\mathrm{SoC}}_{j,\mathrm{BS}}^{t},
\end{gathered}
\label{eq:bess_limits}
\end{equation}
\vspace{-0.5em}
\begin{equation}
\mathrm{SoC}_{j,\mathrm{BS}}^{t+1}
=\mathrm{SoC}_{j,\mathrm{BS}}^{t}
+\eta_{\mathrm{ch}}p_{j,\mathrm{BS}}^{c,t}\Delta t
-\frac{p_{j,\mathrm{BS}}^{d,t}\Delta t}{\eta_{\mathrm{dis}}}.
\label{eq:bess_balance}
\end{equation}
where $p_{j,\mathrm{BS}}^{c,t}$ and $p_{j,\mathrm{BS}}^{d,t}$ denote the charging and discharging powers of battery storage system $j$ at time $t$, respectively. Moreover, $\mathrm{SoC}_{j,\mathrm{BS}}^{t}$ denotes its state of charge, while $\eta_{\mathrm{ch}}$ and $\eta_{\mathrm{dis}}$ denote charging and discharging efficiencies, respectively.

Recently, more advanced battery types have been explored for data center operation, such as ZAB and SFB~\cite{BloomE}. These batteries are different in charging behaviors compared with conventional lithium-ion batteries, as their energy replenishment relies on external reactant replacement or fuel regeneration rather than direct electrochemical recharging~\cite{Zinc}. Therefore, additional constraints are required to characterize their flexibility regions separately from those of conventional BESSs. For a ZAB, introducing a refueling parameter into the SoC evolution equation to represent swap events yields the following constraints:
\begin{equation}
\begin{aligned}
\mathrm{SoC}_{j,\mathrm{ZAB}}^{l}(t+1)
={}&\mathrm{SoC}_{j,\mathrm{ZAB}}^{l}(t)
+\frac{\eta_{\mathrm{ch}}p_{j,\mathrm{ZAB}}^{c,l}(t)\Delta t}
{E_{j,\mathrm{ZAB}}^{l}} \\
&-\frac{p_{j,\mathrm{ZAB}}^{d,l}(t)\Delta t}
{\eta_{\mathrm{dis}}E_{j,\mathrm{ZAB}}^{l}}
+\frac{r_{j,\mathrm{ZAB}}^{l}(t)\Delta t}
{E_{j,\mathrm{ZAB}}^{l}},\\[-1pt]
&\hspace{-3.0cm}
0\le r_{j,\mathrm{ZAB}}^{l}(t)
\le y_{j,\mathrm{ZAB}}^{l}(t)R_{j,\mathrm{ZAB}}^{l},
\quad
0\le y_{j,\mathrm{ZAB}}^{l}(t)\le1 .
\end{aligned}
\label{eq:zab_model}
\end{equation}
where \(r_{j,\mathrm{ZAB}}^{l}(t)\) is the refueling-power rate,
\(R_{j,\mathrm{ZAB}}^{l}\) is its upper bound, and
\(y_{j,\mathrm{ZAB}}^{l}(t)\) is a continuous refueling-activation
variable. 

For the solid fuel battery, there is no electrical charging. Additional operational constraints are imposed as follows:
\begin{equation}
\resizebox{0.96\columnwidth}{!}{$
\begin{gathered}
\mathrm{SoC}_{j,\mathrm{SFB}}^{l}(t+1)
=
\mathrm{SoC}_{j,\mathrm{SFB}}^{l}(t)
-\frac{p_{j,\mathrm{SFB}}^{d,l}(t)\Delta t}
       {\eta_{\mathrm{dis}}E_{j,\mathrm{SFB}}^{l}}
+\frac{f_{j,\mathrm{SFB}}^{l}(t)}
       {E_{j,\mathrm{SFB}}^{l}},\\[4pt]
0\le f_{j,\mathrm{SFB}}^{l}(t)
\le y_{j,\mathrm{SFB}}^{l}(t)F_{j,\mathrm{SFB}}^{l},
\qquad
0\le y_{j,\mathrm{SFB}}^{l}(t)\le 1,
\end{gathered}
$}
\label{eq:sfb_model}
\end{equation}
where \(f_{j,\mathrm{SFB}}^{l}(t)\) is the fuel input,
\(F_{j,\mathrm{SFB}}^{l}\) is its upper bound, and
\(y_{j,\mathrm{SFB}}^{l}(t)\) is a continuous fuel-activation
variable. Electrical charging is excluded by imposing
\(p_{j,\mathrm{SFB}}^{c,l}(t)=0\). The continuous ZAB/SFB activation relaxations preserve the convex LP structure. 

3) Photovoltaic Systems: This representation is a modeling choice used
for evaluating the aggregation method rather than a detailed
inverter-dynamics model. The available PV power is represented as
\begin{equation}
\hat{p}_{i,\mathrm{PV}}^{t}
=
\mathrm{CF}_{i}^{t}P_{i,\mathrm{PV}}^{\mathrm{rated}},
\end{equation}
where $\mathrm{CF}_{i}^{t}$ and
$P_{i,\mathrm{PV}}^{\mathrm{rated}}$ denote the time-varying capacity
factor and rated capacity of PV unit $i$, respectively. Its operational
flexibility is modeled as
\begin{equation}
\resizebox{0.96\columnwidth}{!}{$
\begin{gathered}
p_{i,\mathrm{PV}}^{t}
=(1-\alpha_i^{t})\hat{p}_{i,\mathrm{PV}}^{t},
\qquad \alpha_i^{t}\in[0,1],\\[3pt]
\left|p_{i,\mathrm{PV}}^{t}-p_{i,\mathrm{PV}}^{t-1}\right|
\leq R_{i,\mathrm{PV}}^{\max},
\qquad
\underline{E}_{i,\mathrm{PV}}^{t}
\leq\displaystyle\sum_{\tau=1}^{t}
p_{i,\mathrm{PV}}^{\tau}\Delta t
\leq\overline{E}_{i,\mathrm{PV}}^{t}.
\end{gathered}
$}
\label{eq:pv_model}
\end{equation}
where $\alpha_i^{t}$ denotes the curtailment ratio,
$R_{i,\mathrm{PV}}^{\max}$ denotes the maximum ramp rate between
consecutive time steps, and $\underline{E}_{i,\mathrm{PV}}^{t}$ and
$\overline{E}_{i,\mathrm{PV}}^{t}$ denote the lower and upper bounds on
cumulative PV generation, respectively.

The preceding technology-specific models admit the following unified
convex representation:
\begin{equation}
\mathcal X_i=
\left\{
z_i\in\mathbb R^{n_i}
\;\middle|\;
A_i z_i\le b_i,\;
C_i z_i=d_i
\right\},
\label{eq:unified_device_set}
\end{equation}
where \(z_i\) collects the internal decision variables of DER \(i\),
\(A_i z_i\le b_i\) contains its operating limits, and
\(C_i z_i=d_i\) represents its intertemporal dynamics. Thus,
\(\mathcal X_i\) contains all feasible device trajectories under the
continuous convex model.
\vspace{-1em}
\subsection{Aggregate Power Flexibility Region}
\label{AAA}
Consider a portfolio \(\mathcal I\) of DERs over the scheduling
horizon \(\mathcal T:=\{1,\ldots,T\}\). Because the internal decision
vector \(\boldsymbol z_i\) is not generally identical to the delivered
power, the net-power trajectory of DER \(i\) is defined as
\begin{equation}
\boldsymbol p_i=H_i\boldsymbol z_i,
\qquad
H_i\in\mathbb R^{T\times n_i}.
\label{eq:output_mapping}
\end{equation}
In the absence of cross-DER coupling constraints, the aggregate
flexibility set is the Minkowski sum of the individual flexibility sets:
\begin{equation}
\mathcal P
=
\bigoplus_{i\in\mathcal I}\mathcal P_i
=
\left\{
\boldsymbol p\in\mathbb R^T
\,\middle|\,
\boldsymbol p=\sum_{i\in\mathcal I}\boldsymbol p_i,\;
\boldsymbol p_i\in\mathcal P_i
\right\}.
\label{eq:minkowski_sum}
\end{equation}
For large-scale DER fleets, however, constructing an exact
$H$- or $V$-representation of the Minkowski sum is generally
computationally intractable, as its complexity grows rapidly with the
dimension of the flexibility polytope and the number of DERs. This
challenge is further compounded by cross-DER coupling constraints,
under which the aggregate flexibility region can no longer be expressed
as a simple Minkowski sum of the individual sets.

To formulate the coupled region, define the stacked decision vector as
\(\boldsymbol z:=\operatorname{col}
(\boldsymbol z_i)_{i\in\mathcal I}\). Cross-DER coupling constraints,
including feeder, transformer, and portfolio-level limits, are expressed
as \(A_{\mathrm c}\boldsymbol z\leq\boldsymbol b_{\mathrm c}\).
The resulting feasible aggregate power-trajectory region is
\begin{equation}
\mathcal P:=
\left\{
\boldsymbol p\in\mathbb R^T
\,\middle|\,
\boldsymbol p=\sum_{i\in\mathcal I}H_i\boldsymbol z_i,\;
\boldsymbol z_i\in\mathcal X_i\ \forall i\in\mathcal I,\;
A_{\mathrm c}\boldsymbol z\leq\boldsymbol b_{\mathrm c}
\right\}.
\label{eq:aggregate_region}
\end{equation}
Thus, \(\mathcal P\) is the projection of all locally feasible DER
decisions satisfying the shared coupling constraints onto the aggregate
power space. For every \(i\in\mathcal I\), \(\mathcal X_i\) is defined in
\eqref{eq:unified_device_set}, and its power-flexibility set is the
linear image of \(\mathcal X_i\) under \(H_i\):
\begin{equation}
\mathcal P_i:=H_i\mathcal X_i
=
\left\{
\boldsymbol p_i\in\mathbb R^T
\,\middle|\,
\boldsymbol p_i=H_i\boldsymbol z_i,\;
\boldsymbol z_i\in\mathcal X_i
\right\}.
\label{eq:individual_flexibility}
\end{equation}

To illustrate the conversion of technology-specific constraints into
the unified matrix form, consider BESS \(j\), whose decision vector is
\begin{equation}
\boldsymbol z_{j,\mathrm{BS}}
:=
\begin{bmatrix}
(\boldsymbol p_{j,\mathrm{BS}}^{c})^{\mathsf T}&
(\boldsymbol p_{j,\mathrm{BS}}^{d})^{\mathsf T}&
\boldsymbol s_{j,\mathrm{BS}}^{\mathsf T}
\end{bmatrix}^{\mathsf T}
\in\mathbb R^{3T+1}.
\label{eq:bess_decision_vector}
\end{equation}
Here,
\(\boldsymbol p_{j,\mathrm{BS}}^{c},
\boldsymbol p_{j,\mathrm{BS}}^{d}\in\mathbb R^T\)
collect the charging and discharging powers, respectively, and
\(\boldsymbol s_{j,\mathrm{BS}}\in\mathbb R^{T+1}\) contains the dimensionless SoC trajectory, including the initial
state.

Its charging-power, discharging-power, and SoC limits satisfy
\(A_{j,\mathrm{BS}}\boldsymbol z_{j,\mathrm{BS}}
\leq\boldsymbol b_{j,\mathrm{BS}}\), where
\begin{equation}
\resizebox{0.98\columnwidth}{!}{$
A_{j,\mathrm{BS}}
=
\begin{bmatrix}
 I_T&0_T&0_{T\times(T+1)}\\
-I_T&0_T&0_{T\times(T+1)}\\
0_T&I_T&0_{T\times(T+1)}\\
0_T&-I_T&0_{T\times(T+1)}\\
0_{(T+1)\times T}&0_{(T+1)\times T}&I_{T+1}\\
0_{(T+1)\times T}&0_{(T+1)\times T}&-I_{T+1}
\end{bmatrix},
\quad
\boldsymbol b_{j,\mathrm{BS}}
=
\begin{bmatrix}
\overline{\boldsymbol p}_{j,\mathrm{BS}}^{c}\\
-\underline{\boldsymbol p}_{j,\mathrm{BS}}^{c}\\
\overline{\boldsymbol p}_{j,\mathrm{BS}}^{d}\\
-\underline{\boldsymbol p}_{j,\mathrm{BS}}^{d}\\
\overline{\boldsymbol s}_{j,\mathrm{BS}}\\
-\underline{\boldsymbol s}_{j,\mathrm{BS}}
\end{bmatrix}.
$}
\label{eq:bess_A_b}
\end{equation}
Here, \(I_T\) and \(0_T\) are the \(T\times T\) identity and zero
matrices, respectively, while the barred and underlined vectors collect
the time-dependent upper and lower bounds. Accordingly,
\(A_{j,\mathrm{BS}}\in\mathbb R^{(6T+2)\times(3T+1)}\) and
\(\boldsymbol b_{j,\mathrm{BS}}\in\mathbb R^{6T+2}\).

Define the temporal difference matrix as
\begin{equation}
D_T:=
\begin{bmatrix}
-1&1&0&\cdots&0\\
0&-1&1&\ddots&\vdots\\
\vdots&\ddots&\ddots&\ddots&0\\
0&\cdots&0&-1&1
\end{bmatrix}
\in\mathbb R^{T\times(T+1)}.
\label{eq:temporal_difference_matrix}
\end{equation}
Thus, \(D_T\boldsymbol s_{j,\mathrm{BS}}\) collects the consecutive
SoC differences. The normalized SoC dynamics and initial condition
satisfy
\(C_{j,\mathrm{BS}}\boldsymbol z_{j,\mathrm{BS}}
=\boldsymbol d_{j,\mathrm{BS}}\), where
\begin{equation}
\resizebox{0.98\columnwidth}{!}{$
C_{j,\mathrm{BS}}
=
\begin{bmatrix}
-\dfrac{\eta_{\mathrm{ch}}\Delta t}{E_{j,\mathrm{BS}}}I_T&
 \dfrac{\Delta t}{\eta_{\mathrm{dis}}E_{j,\mathrm{BS}}}I_T&
D_T\\[2mm]
0_{1\times T}&0_{1\times T}&\boldsymbol e_1^{\mathsf T}
\end{bmatrix},
\quad
\boldsymbol d_{j,\mathrm{BS}}
=
\begin{bmatrix}
\boldsymbol 0_T\\
\mathrm{SoC}_{j,\mathrm{BS}}^{\mathrm{init}}
\end{bmatrix}.
$}
\label{eq:bess_C_d}
\end{equation}
The first block row represents the \(T\) SoC-transition equations,
whereas
\(\boldsymbol e_1=[1,0,\ldots,0]^{\mathsf T}\in\mathbb R^{T+1}\)
fixes the initial SoC. Therefore,
\(C_{j,\mathrm{BS}}\in\mathbb R^{(T+1)\times(3T+1)}\) and
\(\boldsymbol d_{j,\mathrm{BS}}\in\mathbb R^{T+1}\).

Under the adopted sign convention, grid consumption is positive;
hence, charging and discharging contribute positively and negatively,
respectively, to the BESS net-power trajectory:
\begin{equation}
\boldsymbol p_{j,\mathrm{BS}}
=H_{j,\mathrm{BS}}\boldsymbol z_{j,\mathrm{BS}},
\qquad
H_{j,\mathrm{BS}}
=
\begin{bmatrix}
I_T&-I_T&0_{T\times(T+1)}
\end{bmatrix}.
\label{eq:bess_output_matrix}
\end{equation}

The EV, ZAB, SFB, and PV matrices are constructed analogously from
their technology-specific constraints. The ZAB and SFB decision vectors
include additional refueling or fuel-input variables, whereas the PV
formulation includes generation and curtailment variables. Their
coefficients are incorporated into \(A_i\), \(C_i\), and \(H_i\)
according to \eqref{eq:unified_device_set}.

The local constraints and output mappings are assembled into a
portfolio-level representation. Let
\(n:=\sum_{i\in\mathcal I}n_i\),
\(m:=\sum_{i\in\mathcal I}m_i\), and
\(q:=\sum_{i\in\mathcal I}q_i\), and define
\begin{equation}
\begin{aligned}
A&:=\operatorname{blkdiag}(A_i)_{i\in\mathcal I},
&\boldsymbol b&:=\operatorname{col}(\boldsymbol b_i)_{i\in\mathcal I},\\
C&:=\operatorname{blkdiag}(C_i)_{i\in\mathcal I},
&\boldsymbol d&:=\operatorname{col}(\boldsymbol d_i)_{i\in\mathcal I},\\
H&:=\begin{bmatrix}H_1&H_2&\cdots&H_{|\mathcal I|}\end{bmatrix}.
\end{aligned}
\label{eq:portfolio_matrices}
\end{equation}
Accordingly,
\(A\in\mathbb R^{m\times n}\),
\(C\in\mathbb R^{q\times n}\), and
\(H\in\mathbb R^{T\times n}\). For the previously defined
\(\boldsymbol z\in\mathbb R^n\), the aggregate trajectory is
$
\boldsymbol p
=\sum_{i\in\mathcal I}\boldsymbol p_i
=H\boldsymbol z
\label{eq:aggregate_mapping}
$.
Let \(m_{\mathrm c}\) denote the number of cross-DER coupling
constraints. Then
\(A_{\mathrm c}\in\mathbb R^{m_{\mathrm c}\times n}\) and
\(\boldsymbol b_{\mathrm c}\in\mathbb R^{m_{\mathrm c}}\), with
\begin{equation}
A_{\mathrm c}
=
\begin{bmatrix}
A_{\mathrm c,1}&A_{\mathrm c,2}&\cdots&
A_{\mathrm c,|\mathcal I|}
\end{bmatrix},
\quad
A_{\mathrm c,i}\in\mathbb R^{m_{\mathrm c}\times n_i}.
\label{eq:coupling_partition}
\end{equation}
Define
\(\overline A:=[A^{\mathsf T}\;
A_{\mathrm c}^{\mathsf T}]^{\mathsf T}\) and
\(\overline{\boldsymbol b}
:=[\boldsymbol b^{\mathsf T}\;
\boldsymbol b_{\mathrm c}^{\mathsf T}]^{\mathsf T}\).
Their detailed block structures are
\begin{equation}
\begingroup
\setlength{\arraycolsep}{3pt}
\renewcommand{\arraystretch}{0.9}
\resizebox{0.9\columnwidth}{!}{$\displaystyle
\overline A =
\begin{bmatrix}
A_1 & 0 & \cdots & 0\\
0 & A_2 & \cdots & 0\\
\vdots & \vdots & \ddots & \vdots\\
0 & 0 & \cdots & A_{|\mathcal I|}\\
A_{\mathrm c,1} & A_{\mathrm c,2} & \cdots
& A_{\mathrm c,|\mathcal I|}
\end{bmatrix},
\quad
\overline{\boldsymbol b} =
\begin{bmatrix}
\boldsymbol b_1\\
\boldsymbol b_2\\
\vdots\\
\boldsymbol b_{|\mathcal I|}\\
\boldsymbol b_{\mathrm c}
\end{bmatrix}.
$}
\endgroup
\label{eq:detailed_coupled_matrices}
\end{equation}
Thus,
\(\overline A\in\mathbb R^{(m+m_{\mathrm c})\times n}\) and
\(\overline{\boldsymbol b}\in\mathbb R^{m+m_{\mathrm c}}\).

The coupled aggregate flexibility set is consequently the projection
of the joint feasible region onto the aggregate power-trajectory space:
\begin{equation}
\mathcal P
=
\left\{
\boldsymbol p\in\mathbb R^T
\,\middle|\,
\boldsymbol p=H\boldsymbol z,\;
\overline A\boldsymbol z\leq\overline{\boldsymbol b},\;
C\boldsymbol z=\boldsymbol d
\right\}.
\label{eq:aggregate_flexibility}
\end{equation}
This set is not generally equal to the Minkowski sum of the individual
flexibility sets because the coupling constraints restrict the DERs
jointly. Moreover, constructing an exact half-space or vertex
representation of this high-dimensional projection can become
computationally prohibitive for large portfolios and long scheduling
horizons. This motivates a scalable approximation method that
characterizes the aggregate flexibility region without enumerating all its facets or vertices.
\vspace{-1em}
\subsection{Proposed Adaptive Directional Support Approximation}
\label{subsec:adaptive_support}
To solve the above existing computational challenge, the aggregate flexibility boundary is sampled through support LPs.
For a unit direction \(\boldsymbol\lambda\in\mathbb R^T\), solve
\begin{equation}
\begin{aligned}
h_{\mathcal P}(\boldsymbol\lambda)
:=\max_{\boldsymbol z}\;&
\boldsymbol\lambda^{\mathsf T}H\boldsymbol z\\[-1mm]
\mathrm{s.t.}\;&
\overline A\boldsymbol z\leq\overline{\boldsymbol b},
\quad C\boldsymbol z=\boldsymbol d,
\end{aligned}
\label{eq:support_lp}
\end{equation}
which returns the feasible boundary trajectory
\(\boldsymbol p^\star(\boldsymbol\lambda)
=H\boldsymbol z^\star(\boldsymbol\lambda)\).
Construction directions are sampled in batches as
\(\boldsymbol\lambda=\boldsymbol g/\|\boldsymbol g\|_2\), where
\(\boldsymbol g\sim\mathcal N(\boldsymbol 0,I_T)\).

Let \(\mathcal S_K\) contain the \(K\) retained boundary trajectories.
Using a separate validation set \(\Lambda_{\mathrm{val}}\), the maximum
normalized support deficit is
\begin{equation}
\varepsilon_K:=
\max_{\boldsymbol\mu\in\Lambda_{\mathrm{val}}}
\frac{
h_{\mathcal P}(\boldsymbol\mu)
-\max_{\boldsymbol p\in\mathcal S_K}
\boldsymbol\mu^{\mathsf T}\boldsymbol p
}{
\max\{|h_{\mathcal P}(\boldsymbol\mu)|,\epsilon_0\}
}.
\label{eq:validation_deficit}
\end{equation}
The true validation supports \(h_{\mathcal P}(\boldsymbol\mu)\) are
computed once and cached. The complete procedure is summarized in
Algorithm~\ref{alg:adaptive_support}. After termination, the retained support points are combined through a
convex-hull construction. The following lemma shows that this operation preserves feasibility and yields an inner approximation of the coupled aggregate region.

\begin{algorithm}[t]
\caption{Adaptive Directional Support Approximation}
\label{alg:adaptive_support}
\small
\begin{algorithmic}[1]
\Require
$\overline A,\overline{\boldsymbol b},C,\boldsymbol d,H$;
batch size $B$; budget $M_{\max}$;
tolerance $\epsilon_{\mathrm{tol}}$;
saturation threshold $\epsilon_{\mathrm{sat}}$;
patience parameter $L$;
validation set $\Lambda_{\mathrm{val}}$;
normalization rule
$\epsilon_0=\max\{10^{-8}\ {\rm kW},
10^{-6}\max_{\boldsymbol\mu\in\Lambda_{\mathrm{val}}}
|h_{\mathcal P}(\boldsymbol\mu)|\}$
\Ensure Feasible inner approximation
$\widehat{\mathcal P}_{\mathrm{in}}$

\State Compute and cache $h_{\mathcal P}(\boldsymbol\mu)$,
$\forall\boldsymbol\mu\in\Lambda_{\mathrm{val}}$
\State Compute $\epsilon_0\gets
\max\{10^{-8}\ {\rm kW},
10^{-6}\max_{\boldsymbol\mu\in\Lambda_{\mathrm{val}}}
|h_{\mathcal P}(\boldsymbol\mu)|\}$
\State $\mathcal S\gets\emptyset,\;
M\gets0,\;
\ell\gets0,\;
\varepsilon_{\mathrm{pre}}\gets+\infty$

\While{$M<M_{\max}$}
    \For{$b=1,\ldots,\min\{B,M_{\max}-M\}$}
        \State Sample $\boldsymbol g\sim\mathcal N(\boldsymbol0,I_T)$
        and set $\boldsymbol\lambda\gets
        \boldsymbol g/\|\boldsymbol g\|_2$
        \State Solve \eqref{eq:support_lp} and obtain
        $\boldsymbol p^\star=H\boldsymbol z^\star$
        \State $\mathcal S\gets
        \mathcal S\cup\{\boldsymbol p^\star\}$;
        $M\gets M+1$
    \EndFor

    \State Evaluate $\varepsilon_M$ using
    \eqref{eq:validation_deficit}

    \If{$\varepsilon_M\leq\epsilon_{\mathrm{tol}}$}
        \State \textbf{break}
    \EndIf

    \If{$|\varepsilon_{\mathrm{pre}}-\varepsilon_M|
        \leq\epsilon_{\mathrm{sat}}$}
        \State $\ell\gets\ell+1$
    \Else
        \State $\ell\gets0$
    \EndIf

    \If{$\ell\geq L$}
        \State \textbf{break}
    \EndIf

    \State $\varepsilon_{\mathrm{pre}}\gets\varepsilon_M$
\EndWhile

\State \Return $\widehat{\mathcal P}_{\mathrm{in}}
\gets\operatorname{conv}(\mathcal S)$
\end{algorithmic}
\end{algorithm}

\begin{lemma}[Feasibility of the inner approximation]
\label{lem:inner_feasibility}
Let $\mathcal P$ be the aggregate power-flexibility region defined
in~\eqref{eq:aggregate_flexibility}. For a finite direction set
$\Lambda_K=\{\boldsymbol\lambda_1,\ldots,\boldsymbol\lambda_K\}$,
let $\boldsymbol z^\star(\boldsymbol\lambda_k)$ be an optimal solution
of~\eqref{eq:support_lp}, and define
\begin{equation}
\boldsymbol p^\star(\boldsymbol\lambda_k)
:=H\boldsymbol z^\star(\boldsymbol\lambda_k),\qquad
\mathcal S_K:=
\{\boldsymbol p^\star(\boldsymbol\lambda_k)\}_{k=1}^{K}.
\label{eq:support_point_set}
\end{equation}
Then every retained support point belongs to $\mathcal P$, and
\begin{equation}
\widehat{\mathcal P}_{\mathrm{in}}
:=\operatorname{conv}(\mathcal S_K)\subseteq\mathcal P.
\label{eq:inner_approximation}
\end{equation}
Thus, the proposed method provides a feasible inner approximation of
the coupled aggregate power-flexibility region.
\end{lemma}

\begin{proof}
For every $\boldsymbol\lambda_k\in\Lambda_K$, the optimizer satisfies
\begin{equation}
\overline A\boldsymbol z^\star(\boldsymbol\lambda_k)
\leq\overline{\boldsymbol b},\qquad
C\boldsymbol z^\star(\boldsymbol\lambda_k)=\boldsymbol d,
\label{eq:support_point_feasibility}
\end{equation}
and hence
$\boldsymbol p^\star(\boldsymbol\lambda_k)
=H\boldsymbol z^\star(\boldsymbol\lambda_k)\in\mathcal P$.
For any
$\widehat{\boldsymbol p}\in\operatorname{conv}(\mathcal S_K)$,
there exist $\theta_k\geq0$ with $\sum_{k=1}^{K}\theta_k=1$ such that
\begin{equation}
\widehat{\boldsymbol p}
=\sum_{k=1}^{K}\theta_k
\boldsymbol p^\star(\boldsymbol\lambda_k),\qquad
\widehat{\boldsymbol z}
:=\sum_{k=1}^{K}\theta_k
\boldsymbol z^\star(\boldsymbol\lambda_k).
\label{eq:convex_combination}
\end{equation}
By linearity,
\begin{equation}
H\widehat{\boldsymbol z}=\widehat{\boldsymbol p},\qquad
\overline A\widehat{\boldsymbol z}
\leq\overline{\boldsymbol b},\qquad
C\widehat{\boldsymbol z}=\boldsymbol d.
\label{eq:convex_combination_feasibility}
\end{equation}
Hence, $\widehat{\boldsymbol p}\in\mathcal P$, proving
$\operatorname{conv}(\mathcal S_K)\subseteq\mathcal P$.
\end{proof}

\section{Theoretical Analysis of Computational Complexity}
\label{sec:complexity}

\begingroup
\setlength{\abovedisplayskip}{4pt}
\setlength{\belowdisplayskip}{4pt}
\setlength{\abovedisplayshortskip}{3pt}
\setlength{\belowdisplayshortskip}{3pt}
\setlength{\jot}{2pt}

This section provides a theoretically grounded analysis of the principal
factors affecting the computational complexity of the proposed
aggregation method. The computational burden of the principal components
is examined, and an expression for the overall complexity is derived.
According to the proposed models, the number of decision variables and
the quantity of equality and inequality constraints play critical roles
in determining the computational burden. The proposed method solves a
sequence of LPs, one for each sampled projection direction. The
computational cost also depends on the algorithm and implementation used
by the selected solver.

Based on the literature and numerical experience, an interior-point
method (IPM) is commonly used to solve large-scale sparse LP
problems~\cite{andersen1996implementation}. The computational cost of an
IPM depends on two principal components: the number of Newton iterations
required for convergence and the computational effort per iteration.
Under a short-step path-following scheme, the number of Newton iterations
required to attain an optimality tolerance $\varepsilon$ satisfies
$K_{\mathrm{IPM}}
=\mathcal{O}\!\left(
\sqrt{\vartheta}\log(\vartheta/\varepsilon)
\right)$, where $K_{\mathrm{IPM}}$, $\vartheta$, and $\varepsilon$
denote the number of IPM iterations, the barrier parameter, and the
optimality tolerance, respectively. Therefore, the computational burden
depends on the barrier parameter and the prescribed solution accuracy.
We formalize this result in Lemma~\ref{lemma2}.

For the second part, the computational complexity per iteration of the
IPM is primarily determined by the solution of the Newton system, which
includes factoring the matrix, solving equations, computing, and
updating variables. The cost of solving this linear system depends
critically on the structure and sparsity of the constraint matrix $A$.
Based on Lemma~\ref{lemma3}, the Newton-step cost depends on the nonzeros
and fill-in of the KKT or normal-equation factorization, which itself
depends on the sparsity pattern of the matrix $A$. If $A$ is a dense
matrix, the complexity is a quadratic formula
\textbf{($\mathcal{O}(\operatorname{nnz}(A))$)} related to the number
of nonzero entries. The more sparse $A$ is, the less computational time
is needed~\cite{luce2014minimum}.

In this paper, the matrix $A$ from the proposed LP-based flexibility
projection model exhibits a high degree of sparsity because only a small
subset of decision variables is involved in each constraint of the
model. For instance, at each time step, the power balance constraints
are only related to specific time variables, while operational bounds
and charging limits introduce mostly diagonal or block-diagonal
constraints. Moreover, aggregator-level coupling constraints and SoC
dynamic constraints also affect only a limited subset of variables, not
all variables simultaneously. As a result, the majority of entries in
$A$ are zero, and only a few nonzero elements appear per row or column.
Therefore, the computational complexity of each iteration and the total
complexity can be represented as
\begin{equation}
\mathcal{O}\!\left(\operatorname{nnz}(A)\right) = \mathcal{O}(n_z),
\qquad
\mathcal{C}_{\mathrm{tot}} = M K_{\mathrm{IPM}}
\mathcal{O}\!\left(\operatorname{nnz}(A)\right).
\label{eq:total1}
\end{equation}

\begin{lemma}
\label{lemma2}
Under a short-step path-following scheme using a logarithmic barrier,
the number of Newton steps $K_{\mathrm{IPM}}$ required to achieve a
duality gap
$\mu(t_{K_{\mathrm{IPM}}})\leq\varepsilon$ satisfies
\begin{equation}
K_{\mathrm{IPM}}
\leq C_1\sqrt{\vartheta}
\log\!\left(\frac{\mu(t_0)}{\varepsilon}\right),
\label{eq:lemma2_bound}
\end{equation}
where $C_1>0$ depends only on the selected short-step safeguard. For
the standard logarithmic barrier associated with $m_{\mathrm{in}}$
inequality constraints,
\begin{equation}
K_{\mathrm{IPM}}
=\mathcal{O}\!\left(
\sqrt{m_{\mathrm{in}}}
\log\!\frac{m_{\mathrm{in}}}{\varepsilon}
\right).
\label{eq:lemma2_asymptotic}
\end{equation}
\end{lemma}

\begin{proof}
For a logarithmic barrier with parameter $\vartheta$, let $t_k$ denote
the path-following parameter at iteration $k$. Under the short-step
scheme, the parameter is updated according to
\begin{equation}
t_{k+1}
=\left(1+\frac{\sigma}{\sqrt{\vartheta}}\right)t_k,
\label{eq:path_parameter_update}
\end{equation}
where $\sigma\in(0,1)$ is selected such that a bounded number of Newton
corrections restores the prescribed central-neighborhood condition
\cite{nesterov1994interior,wright1997primal}.

Since the duality gap along the central path satisfies
$\mu(t)=\vartheta/t$, the update in
\eqref{eq:path_parameter_update} yields
\begin{equation}
\begin{aligned}
\mu(t_{k+1})
=\frac{\vartheta}{t_{k+1}}
=\frac{\mu(t_k)}{1+\sigma/\sqrt{\vartheta}}
\leq
\exp\!\left(-\frac{\sigma}{\sqrt{\vartheta}}\right)
\mu(t_k).
\end{aligned}
\label{eq:duality_gap}
\end{equation}
After $K_{\mathrm{IPM}}$ iterations,
\begin{equation}
\mu(t_{K_{\mathrm{IPM}}})
\leq
\mu(t_0)
\exp\!\left(
-\frac{\sigma K_{\mathrm{IPM}}}{\sqrt{\vartheta}}
\right).
\label{eq:duality_gap_iterations}
\end{equation}
Therefore, the condition
$\mu(t_{K_{\mathrm{IPM}}})\leq\varepsilon$ is satisfied when
\begin{equation}
K_{\mathrm{IPM}}
\geq
\frac{\sqrt{\vartheta}}{\sigma}
\log\!\left(\frac{\mu(t_0)}{\varepsilon}\right).
\label{eq:required_ipm_iterations}
\end{equation}
Absorbing $\sigma^{-1}$ and the bounded number of safeguard corrections
into $C_1$ gives \eqref{eq:lemma2_bound}. For a logarithmic barrier over
$m_{\mathrm{in}}$ inequality slacks,
$\vartheta=\mathcal{O}(m_{\mathrm{in}})$, which gives
\eqref{eq:lemma2_asymptotic}.
\end{proof}

The result in Lemma~\ref{lemma2} characterizes the number of IPM
iterations. The computational cost of each iteration is determined by
the Newton system, as described next.

\begin{lemma}
\label{lemma3}
Suppose that the device model is local in time, such that each
time-indexed device variable interacts with only
$\mathcal{O}(1)$ neighboring intervals, and that the cross-DER
constraints remain sparse. Let $\mathcal{K}_k$ denote the Newton-system
matrix at iteration $k$. If a fill-reducing ordering yields a uniformly
bounded elimination width, then each Newton step can be computed in
$\mathcal{O}(NT)$ time. For a fixed scheduling horizon, the
per-iteration complexity reduces to $\mathcal{O}(N)$.
\end{lemma}

\begin{proof}
For a given projection direction $\lambda$, the primal support problem
can be written as
\begin{equation}
\begin{aligned}
\max_{\boldsymbol z,\boldsymbol s}\quad
&\lambda^{\top}H\boldsymbol z\\
\mathrm{s.t.}\quad
&\bar A\boldsymbol z+\boldsymbol s=\bar{\boldsymbol b},
\quad C\boldsymbol z=\boldsymbol d,
\quad \boldsymbol s\geq\boldsymbol 0,
\end{aligned}
\label{eq:complexity_primal}
\end{equation}
where $\boldsymbol s$ is the inequality-slack vector. The corresponding
dual problem is
\begin{equation}
\begin{aligned}
\min_{\boldsymbol y,\boldsymbol\xi}\quad
&\bar{\boldsymbol b}^{\top}\boldsymbol y
+\boldsymbol d^{\top}\boldsymbol\xi\\
\mathrm{s.t.}\quad
&\bar A^{\top}\boldsymbol y+C^{\top}\boldsymbol\xi
=H^{\top}\lambda,
\quad \boldsymbol y\geq\boldsymbol 0,
\end{aligned}
\label{eq:complexity_dual}
\end{equation}
where $\boldsymbol y$ and $\boldsymbol\xi$ are the inequality- and
equality-constraint multipliers, respectively. Define:
\begin{equation}
Y:=\operatorname{diag}(\boldsymbol y),
\qquad
\mathcal S:=\operatorname{diag}(\boldsymbol s).
\label{eq:diagonal_matrices}
\end{equation}
The perturbed KKT conditions are
\begin{equation}
\begin{aligned}
\bar A\boldsymbol z+\boldsymbol s
&=\bar{\boldsymbol b},&
C\boldsymbol z&=\boldsymbol d,\\
\bar A^{\top}\boldsymbol y+C^{\top}\boldsymbol\xi
&=H^{\top}\lambda,&
Y\boldsymbol s&=\mu\boldsymbol 1.
\end{aligned}
\label{eq:perturbed_kkt}
\end{equation}
where $\mu>0$ is the complementarity parameter.

For the current iterate
$(\boldsymbol z,\boldsymbol s,\boldsymbol y,\boldsymbol\xi)$, define
the residuals as
\begin{equation}
{\small
\begin{alignedat}{2}
\boldsymbol r_p
&:=\bar{\boldsymbol b}-\bar A\boldsymbol z-\boldsymbol s,
&\qquad
\boldsymbol r_e
&:=\boldsymbol d-C\boldsymbol z,\\
\boldsymbol r_d
&:=H^{\top}\lambda-\bar A^{\top}\boldsymbol y
-C^{\top}\boldsymbol\xi,
&\qquad
\boldsymbol r_\mu
&:=\mu\boldsymbol 1-Y\boldsymbol s.
\end{alignedat}
}
\label{eq:kkt_residuals}
\end{equation}
The Newton corrections
$(\Delta\boldsymbol z,\Delta\boldsymbol s,
\Delta\boldsymbol y,\Delta\boldsymbol\xi)$
satisfy
\begin{equation}
{\small
\begin{alignedat}{2}
\bar A\Delta\boldsymbol z+\Delta\boldsymbol s
&=\boldsymbol r_p,
&\qquad
C\Delta\boldsymbol z
&=\boldsymbol r_e,\\
-\bar A^{\top}\Delta\boldsymbol y
-C^{\top}\Delta\boldsymbol\xi
&=\boldsymbol r_d,
&\qquad
Y\Delta\boldsymbol s+\mathcal S\Delta\boldsymbol y
&=\boldsymbol r_\mu.
\end{alignedat}
}
\label{eq:newton_equations}
\end{equation}

Rather than assuming that $C$ is invertible, the coupled Newton
equations are solved directly or reduced through a valid block
elimination to $\mathcal K_k\Delta\boldsymbol u_k=\boldsymbol r_k,
\label{eq:newton_system}$
where $\mathcal K_k$ is the sparse KKT or normal-equation matrix. Let $P^{\top}\mathcal K_kP=LDL^{\top}
\label{eq:ldl_factorization}$
denotes a sparse symmetric factorization under a fill-reducing
permutation $P$. If $\ell_j$ denotes the number of nonzero entries in
column $j$ of $L$, the factorization cost is bounded by $\mathcal{O}\!\left(
\sum_{j=1}^{n_{\mathcal K}}\ell_j^2
\right),
\label{eq:factorization_cost}$, where $n_{\mathcal K}:=\dim(\mathcal K_k)$. The subsequent triangular
solutions require
$\mathcal{O}(\operatorname{nnz}(L))$ operations.

Let $w:=\max_j\ell_j$ denote the elimination width. Then
\begin{equation}
\sum_{j=1}^{n_{\mathcal K}}\ell_j^2
\leq w^2n_{\mathcal K},
\qquad
\operatorname{nnz}(L)\leq wn_{\mathcal K}.
\label{eq:bounded_fill}
\end{equation}
When the elimination width satisfies $w=\mathcal{O}(1)$ and each DER
contributes a fixed number of variables and constraints per interval,
$n_{\mathcal K}=\mathcal{O}(NT)$. Therefore, the factorization and
triangular-solution costs are $\mathcal{O}(NT)$. For fixed $T$, this
reduces to $\mathcal{O}(N)$, which proves
Lemma~\ref{lemma3}.
\end{proof}

\endgroup

\section{Numerical Simulation}
In this section, two cases are proposed to verify and evaluate the performance of the LP-based projection method. One is a small case including 20 DER devices 
aggregation at 6 time steps to illustrate the accuracy and validity. Another is a large-scale case, which involves a large-scale DER and sensitive analysis, indicating the proposed approximation method's performance in terms of computability and applicability. 
To better reflect real-world conditions and demonstrate the generality of the proposed method, DERs are modeled as heterogeneous. 

\begin{table*}[t]
\centering
\caption{Large-scale Simulation settings}
\scriptsize
\vspace{-1em}
\begin{tabular}{ccccccccccc}
\hline
\makecell{DERs} & \makecell{Maximum \\ Charging (kw)} & \makecell{Maximum \\ Discharging (kw)} & \makecell{Minimum\\ Energy (kWh)} & \makecell{Maximum \\ Energy (kWh)} &
\makecell{Max. Refuel\\ Power (kw)} &
\makecell{Max. Fuel \\ Added (kwh)} & 
\makecell{Max. Peak\\ Power (kw)} &
\makecell{Max. Ramp\\ Rate (kw)} \\
\hline
EVs (400) & $\mathcal{U}(5, 10)$ & $\mathcal{U}(5, 10)$ & 10 & 30 & - & - & - & -
\\
\hline
LIB (250) & $\mathcal{U}(5, 10)$ & $\mathcal{U}(5, 10)$ & - & - & - & -& - & -\\
\hline
ZAB (150) & $\mathcal{U}(5, 10)$ & $\mathcal{U}(5, 10)$ & - & - & $\mathcal{U}(3, 8)$ & - & - & - \\
\hline
SFB (100) & - & $\mathcal{U}(5, 10)$ & - & - & - & $\mathcal{U}(3, 10)$ & - & - \\
\hline
PVs (100) & - & - & - & - & - & - & $\mathcal{U}(6, 16)$ & $\mathcal{U}\!\left(0.3P_{\text{peak}},\,0.5P_{\text{peak}}\right)$ \\
\hline
\end{tabular}
\label{Fig:largesimulation}
\end{table*}
\vspace{-1em}
\subsection{Small-scale Simulation Case}
The small-scale portfolio comprises 20 DERs: 8 EVs, 5 LIBs,
3 ZABs, 2 SFBs, and 2 PVs, operated over $T=6$ time steps
with $\Delta t=1$~h. Charging and discharging power limits
for applicable storage devices are uniformly sampled from
$[5,10]$~kW, with efficiencies $\eta_{ch}=\eta_{dis}=0.95$.
Non-PV energy requirements are randomly assigned within
$[10,30]$~kWh. Refueling power and fuel capacity limits are
uniformly sampled from $[3,8]$~kW and $[3,10]$~kWh,
respectively. PV ramp limits are uniformly sampled between
$0.3$ and $0.5$ times the PV peak power. The grid power
and aggregate energy limits are $25$~kW and $40$~kWh,
respectively, and the maximum direction budget is
$M_{\max}=96$.
\vspace{-1em}
\subsection{A Large-scale Simulation Case}
The large-scale case evaluates the scalability of the proposed LP-based directional support framework for up to $2000$ DERs. Device-level parameters follow the small-scale case and are summarized in Table~II. The horizon remains $T=6$ with $\Delta t=1$~h, while the large-scale feeder power and aggregate energy limits are $1000$~kW and $3750$~kWh, respectively. Runtime experiments cover selected portfolio sizes from $N=20$ to $2000$ and direction budgets from $M=1$ to $256$ to assess model size, construction time, per-direction solve time, and total runtime. These experiments complement the accuracy analysis and assess empirical scaling.

\begin{figure*}[t]
\centering
\subfigure[\textbf{t=1, N=20}]
{
    \begin{minipage}[b]{.31\linewidth}
        \centering
        \label{poweraggregation}\includegraphics[scale=0.24]{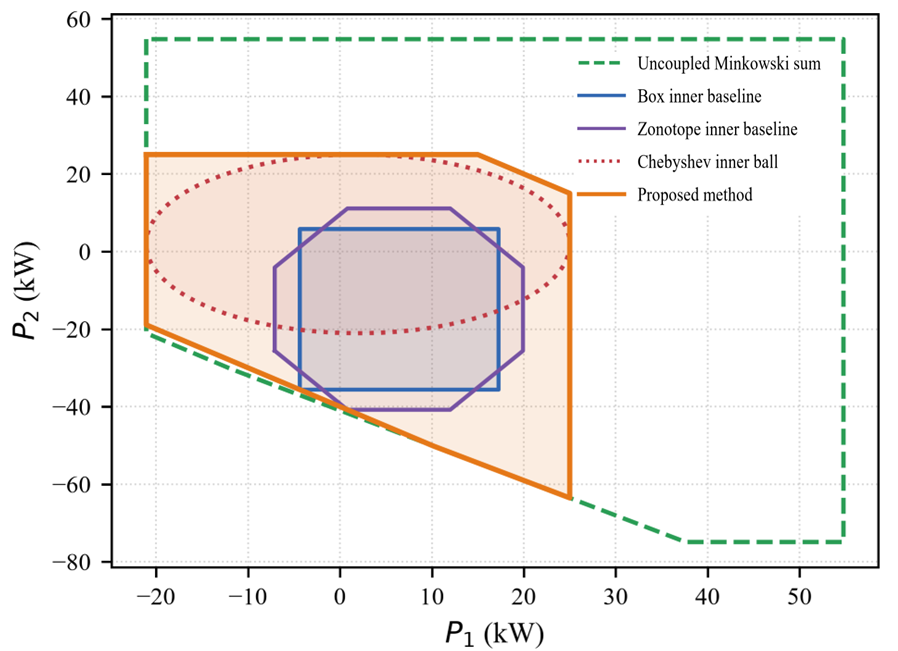}
    \end{minipage}
}
\subfigure[\textbf{t=2, N=20}]
{
 	\begin{minipage}[b]{.31\linewidth}
        \centering
        \label{poweraggregation2}\includegraphics[scale=0.5]{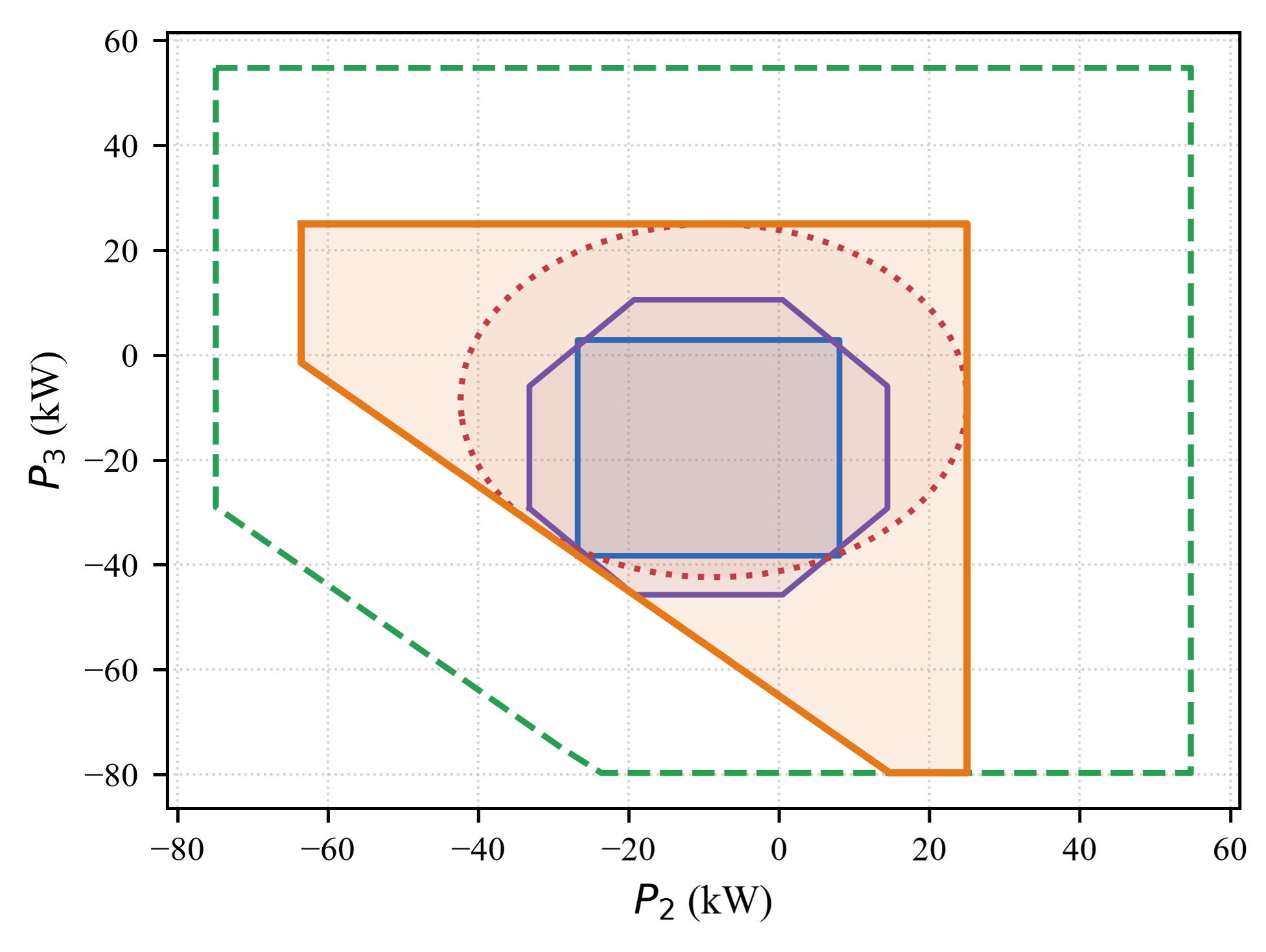}
    \end{minipage}
}
\subfigure[\textbf{t=3, N=20}]
{
 	\begin{minipage}[b]{.31\linewidth}
        \centering
        \label{poweraggregation3}\includegraphics[scale=0.5]{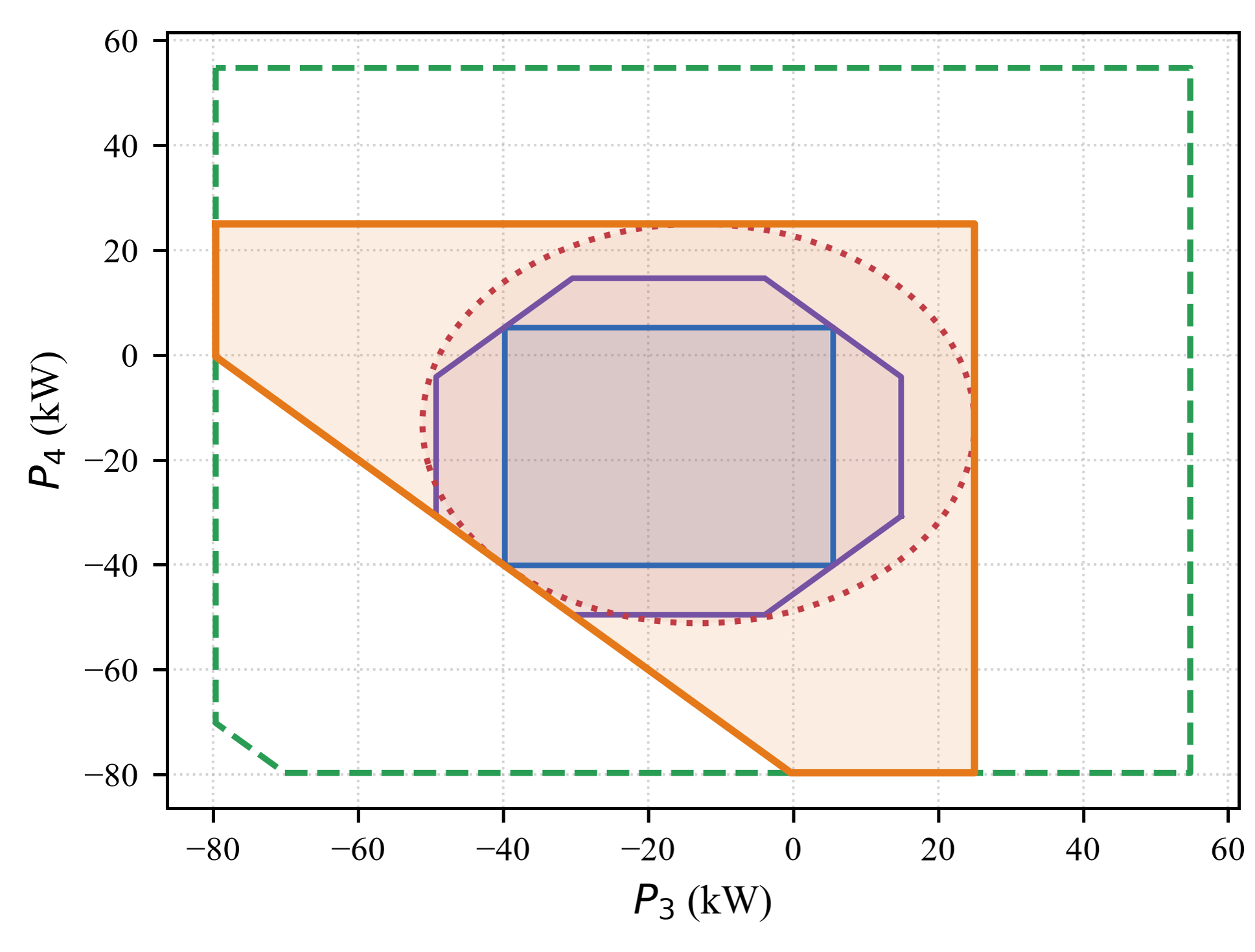}
    \end{minipage}
}
  \vspace{-.5em}
\caption{Comparison of projected aggregate net-power flexibility regions
for the small-scale case with $N=20$ DERs at selected consecutive time
steps. The uncoupled Minkowski sum is a reference envelope, while the
remaining models are inner approximations of the coupled flexibility
region}
\end{figure*}

\vspace{-1em}
\subsection{Comparison with Existing Aggregate Models}
\label{subsec:benchmark_comparison}
The small-scale case is used to compare the proposed method with four 
aggregate flexibility models: an uncoupled Minkowski-sum model, a box 
inner approximation, a zonotope inner approximation, and a Chebyshev 
inner-ball approximation. All methods use the same DER parameters, 
scheduling horizon, and local operating constraints. The box, zonotope, 
inner-ball, and proposed models additionally enforce the same feeder and 
portfolio-level coupling constraints.

The uncoupled flexibility region is calculated as Eq.~(\ref{eq:minkowski_sum}) above.
The box benchmark is represented by
\begin{equation}
\mathcal{P}_{\mathrm{box}}
=
\left\{
\boldsymbol c_{\mathrm b}
+
\operatorname{diag}(\boldsymbol r_{\mathrm b})
\boldsymbol u
\,\middle|\,
\|\boldsymbol u\|_{\infty}\leq 1
\right\},
\label{eq:box_benchmark}
\end{equation}
where $\boldsymbol c_{\mathrm b}$ and $\boldsymbol r_{\mathrm b}$ denote
the center and half-width vector, respectively. For each two-dimensional
projection, the box is homothetically scaled until its four vertices are
feasible under the coupled DER model.

The zonotope benchmark is defined as
\begin{equation}
\mathcal{P}_{\mathrm{zono}}
=
\left\{
\boldsymbol c_{\mathrm z}
+
G_{\mathrm z}\boldsymbol u
\,\middle|\,
\|\boldsymbol u\|_{\infty}\leq 1
\right\},
\label{eq:zonotope_benchmark}
\end{equation}
where $\boldsymbol c_{\mathrm z}$ is the zonotope center and
$G_{\mathrm z}$ contains its generators. Four generator orientations,
$0^{\circ}$, $45^{\circ}$, $90^{\circ}$, and $135^{\circ}$, are fixed
before the comparison. The zonotope is uniformly contracted until all
of its vertices satisfy the coupled portfolio constraints. 

Motivated by the Chebyshev-centering models in
\cite{10738509,11225608}, the inner-ball
benchmark is expressed as
\begin{equation}
\mathcal{P}_{\mathrm{ball}}
=
\left\{
\boldsymbol p
\,\middle|\,
\|\boldsymbol p-\boldsymbol c_{\mathrm o}\|_2
\leq r_{\mathrm o}
\right\},
\label{eq:inner_ball_benchmark}
\end{equation}
where $\boldsymbol c_{\mathrm o}$ and $r_{\mathrm o}$ are the center and
radius, respectively. Given a certified inner polytope
$\{\boldsymbol p\mid F\boldsymbol p\leq\boldsymbol g\}$, its largest
inscribed Euclidean ball is obtained from
\begin{equation}
\begin{aligned}
\max_{\boldsymbol c_{\mathrm o},\,r_{\mathrm o}}\quad
&r_{\mathrm o}\\
\mathrm{s.t.}\quad
&F_j\boldsymbol c_{\mathrm o}
+\|F_j\|_2 r_{\mathrm o}
\leq g_j,
\quad \forall j.
\end{aligned}
\label{eq:chebyshev_ball}
\end{equation}
The resulting ball is therefore contained in the coupled flexibility
region. 

For comparison, the area-coverage ratio of an inner model
$\widehat{\mathcal P}$ is defined as
\begin{equation}
\rho_{\mathrm{area}}
=
\frac{
\operatorname{Area}(\widehat{\mathcal P})
}{
\operatorname{Area}(\mathcal P_{\mathrm{ref}})
}
\times 100\%,
\label{eq:area_coverage}
\end{equation}
where $\mathcal P_{\mathrm{ref}}$ is obtained using 180 uniformly
distributed validation directions. 

Fig.{2} compares the five models for the
consecutive-time projections $(P_t,P_{t+1})$. The uncoupled Minkowski
sum is substantially larger than the coupled region because it does not
enforce the feeder power and aggregate energy limits. Its projected area
is approximately $2.31$--$2.74$ times that of the coupled reference,
demonstrating that the simple Minkowski sum can considerably overestimate
deliverable portfolio flexibility. The box model captures $23.10\%$--$29.63\%$ of the coupled projected
area, with an average coverage of $26.37\%$. The zonotope improves the
average coverage to $41.05\%$, with projection-level values ranging from
$35.76\%$ to $43.72\%$. The Chebyshev inner ball provides an average
coverage of $57.50\%$, ranging from $55.12\%$ to $58.32\%$. Although
these compact models are computationally convenient, their prescribed
geometries cannot fully represent the asymmetric and temporally coupled
boundaries of the portfolio flexibility region. In contrast, the proposed directional approximation matches the
180-direction reference within numerical tolerance for all evaluated
two-dimensional projections. Its area-coverage ratio is about $100\%$ in
these projections, using between 16 and 64 construction directions.
This result does not imply that the entire $T$-dimensional flexibility
region is recovered exactly. 

\begin{table}[t]
\centering
\footnotesize
\caption{Computational Time Comparison with Grid Limits (T=6)}
\label{tab:runtime}
\renewcommand{\arraystretch}{1.1}
\setlength{\tabcolsep}{5pt}
\vspace{-1em}
\begin{tabular}{ccccccc}
\hline
$N$ & $M$ & Vars & Cons & Build (s) & Median Solve (s) & Total (s) \\ \hline
\multicolumn{7}{l}{\textit{(a) Small case}} \\
20 & 24  & 734 & 1072  & 0.046 & 0.007 & 0.198 \\
20 & 48  & 734 & 1072 & 0.046 & 0.005 & 0.282 \\
20 & 64  & 734 & 1072 & 0.046 &  0.008 & 0.563 \\
20 & 128 & 734 & 1072 & 0.046 & 0.008 & 0.987 \\[1pt]
100 & 24  & 3586 & 5315 & 0.202 & 0.078 & 2.047 \\ 
100 & 48 & 3586 & 5315 & 0.202 & 0.085 & 4.776 \\
100 & 64 & 3586 & 5315 & 0.202 & 0.073 & 5.068\\
100 & 128 & 3586 & 5315 & 0.202 & 0.055 & 8.070 \\
100 & 256 & 3586 & 5315 & 0.202 & 0.059 & 16.230 \\ 
200 & 24 & 7214 & 10453 & 0.358 & 0.380 & 9.499 \\
200 & 48 & 7214 & 10453 & 0.358 & 0.298 & 15.378\\
200 & 64 & 7214 & 10453 & 0.358 & 0.315 & 19.658\\
200 & 128 & 7214 & 10453 & 0.358 & 0.266 & 35.614\\
200 & 256 & 7214 & 10453 & 0.358 & 0.277 & 72.436\\
\hline
\hline
\multicolumn{7}{l}{\textit{(b) Large case}} \\
500 & 24 & 17822 & 26065 & 0.874 & 1.798 & 38.251 \\ 
500 & 48 & 17822 & 26065 & 0.874 & 1.256 & 59.138 \\
500 & 64 & 17822 & 26065 & 0.874 & 1.599 & 92.584\\
500 & 128 & 17822 & 26065 & 0.874 & 1.354 & 177.066 \\
500 & 256 & 17822 & 26065 & 0.874 & 1.434 & 360.863\\ 
1000 & 24 & 35770 & 51584 & 1.837 & 3.719 & 83.717 \\
1000 & 48 & 35770 & 51584 & 1.837 & 3.102 & 144.305 \\
1000 & 64 & 35770 & 51584 & 1.837 & 3.437 & 196.714\\
1000 & 128 & 35770 & 51584 & 1.837 & 3.118 & 368.852\\
1000 & 256 & 35770 & 51584 & 1.837 & 3.273 & 801.275\\
1500 & 24 & 53706 & 77277 & 2.843 & 5.213 & 117.3211 \\
1500 & 48 & 53706 & 77277 & 2.843 & 4.732 & 204.352 \\
1500 & 64 & 53706 & 77277 & 2.843 & 5.183 & 298.677\\
1500 & 128 & 53706 & 77277 & 2.843 & 5.372 & 635.823\\
1500 & 256 & 53706 & 77277 & 2.843 & 5.204 & 1170.138\\
2000 & 24 & 71558 & 103315 & 4.214 & 6.691 & 155.7275\\
2000 & 48 & 71558 & 103315 & 4.214 & 7.033 & 327.785 \\
2000 & 64 & 71558 & 103315 & 4.214 & 6.361 & 359.339\\
2000 & 128 & 71558 & 103315 & 4.214 & 7.451 & 854.071\\
2000 & 256 & 71558 & 103315 & 4.214 & 6.162 & 1400.095\\
\hline
\end{tabular}
\end{table}
\vspace{-1em}
\subsection{Direction-Selection Sensitivity}
\label{subsubsec:M_sensitivity}

The effect of the direction budget $M$ is evaluated using the nested
sequence $M\in\{16,24,32,48,64,96\}$.
For each projection, normalized Gaussian construction directions are
introduced in batches of eight. Each sampled direction requires one
support LP and generates one feasible boundary trajectory. A separate
set of 180 uniformly distributed directions is used to evaluate the
maximum normalized support deficit $\varepsilon_M$. The validation
supports are calculated once and cached to ensure that all values of
$M$ are evaluated against the same reference.

As shown in Fig.~\ref{fig:M_sensitivity}, increasing $M$ generally
reduces the support deficit, although the reduction is not necessarily
smooth. The mean deficit decreases from $1.56\times10^{-1}$ at $M=16$
to $4.64\times10^{-3}$ at $M=24$. It remains nearly unchanged for
$M=32$ and $M=48$ before decreasing to numerical tolerance at $M=64$.
No measurable improvement is obtained by increasing the budget from
64 to 96.

The intermediate plateau is a consequence of the polyhedral geometry
and directional sampling. Additional directions do not modify the
current approximation when their support solutions coincide with
vertices that have already been identified. The approximation improves
only when a newly sampled direction exposes a previously unidentified
facet or extreme point. Consequently, the error curve can exhibit flat
segments followed by abrupt decreases. Such behavior is expected for a
support-function-based approximation and does not necessarily indicate
algorithmic stagnation.

For the considered small case, $M=64$ represents the empirical accuracy
elbow because all five projections satisfy the prescribed tolerance by
this point. Nevertheless, $M_{\max}=96$ is retained as a conservative
upper budget. During execution, sampling is terminated earlier if $\varepsilon_M\leq\epsilon_{\mathrm{tol}}$. If the improvement remains below the saturation threshold for six
consecutive batches, or if $M_{\max}$ is reached. With
$\epsilon_{\mathrm{tol}}=0.02$, the individual projections terminate
after 16--64 directions, resulting in an average of 30.4 support LPs
per projection.

\begin{figure}[t]
\centering
\includegraphics[width=0.86\columnwidth]
{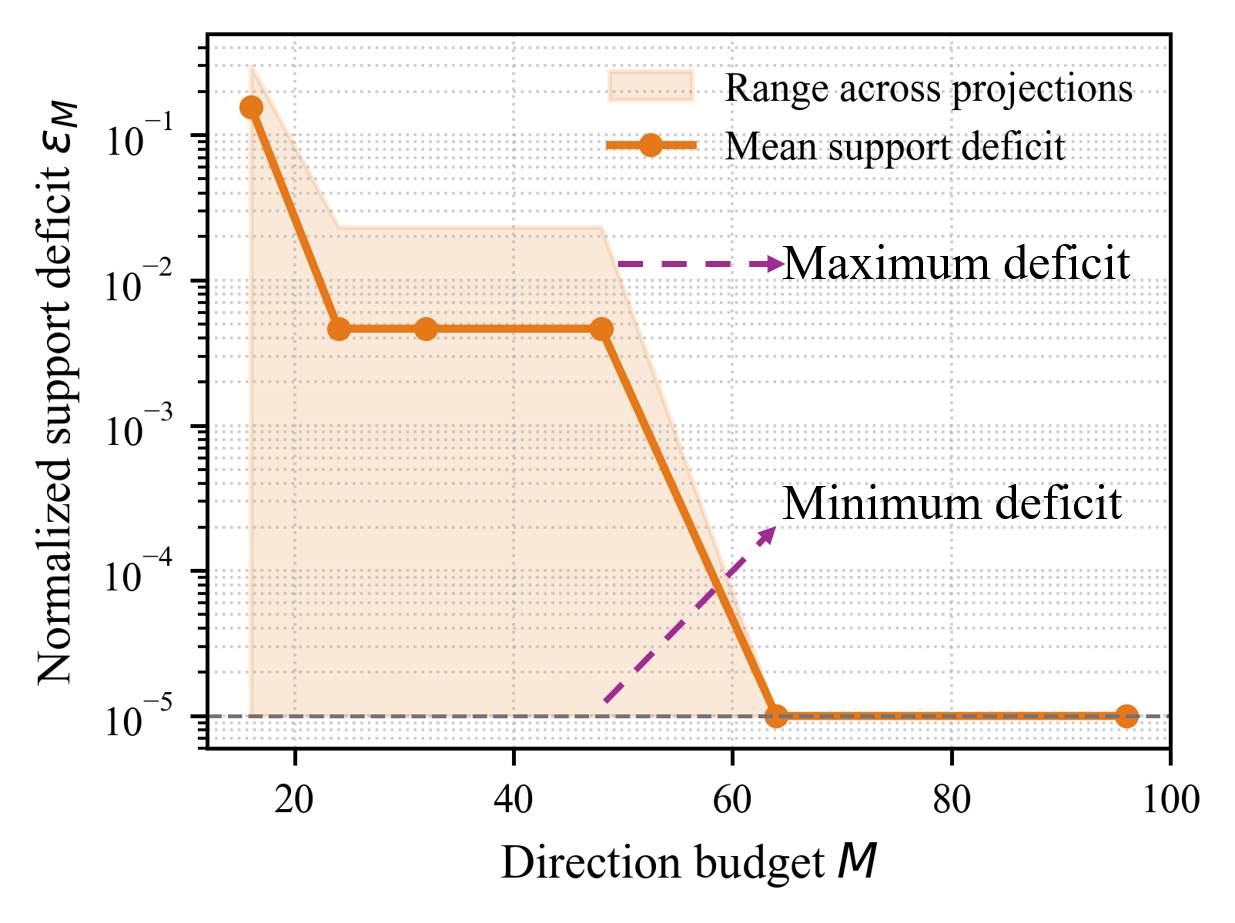}
\vspace{-1em}
\caption{Sensitivity of the normalized support deficit to the direction
budget $M$. The solid line shows the mean over the five consecutive-time
projections, while the shaded region indicates the corresponding
minimum-to-maximum range. Values below $10^{-5}$ are displayed at
$10^{-5}$ for visualization on the logarithmic scale.}
\label{fig:M_sensitivity}
\end{figure}

\subsection{Large-Scale Scalability Analysis}
Table~\ref{tab:runtime} reports the total runtime and median
per-direction solve time for different portfolio sizes $N$ and
direction budgets $M$. As shown in Fig.~\ref{fig:runtime_vsN_M128},
both measures increase regularly with $N$ for a fixed budget of
$M=128$. For the grid-constrained formulation, the total runtime
increases from $368.852$~s at $N=1000$ to $635.823$~s at $N=1500$
and $854.071$~s at $N=2000$. The observed trend is approximately
linear over the tested range and is consistent with the growth of
the support-LP dimensions discussed in Lemmas~\ref{lemma2} and
\ref{lemma3}. 

Fig.~\ref{fig:runtime_vsM_N2000} shows the effect of $M$ for
$N=2000$. The total runtime increases approximately linearly with
$M$, whereas the median per-direction solve time remains nearly
constant for moderate and large values of $M$. Thus, the runtime
growth is primarily caused by the increasing number of support
evaluations. The uncoupled Minkowski-sum model has the lowest per-direction runtime
because its support evaluation is separable across DERs. However, it
neglects cross-DER coupling and therefore serves only as an uncoupled
reference envelope. Among the coupled formulations, the grid-constrained
case is faster in the reported experiments, possibly because its tighter
bounds facilitate solver presolve, although this advantage is
instance-dependent.

Overall, the proposed directional method remains computationally
tractable for portfolios of up to $2000$ DERs. Runtime scales
approximately linearly with $M$, while its dependence on $N$ is governed
by the size and sparsity of each support LP. Because directional
evaluations are independent, parallel computation can further reduce
the wall-clock time.

\begin{figure}[t]
  \centering
  \includegraphics[width=0.86\columnwidth]{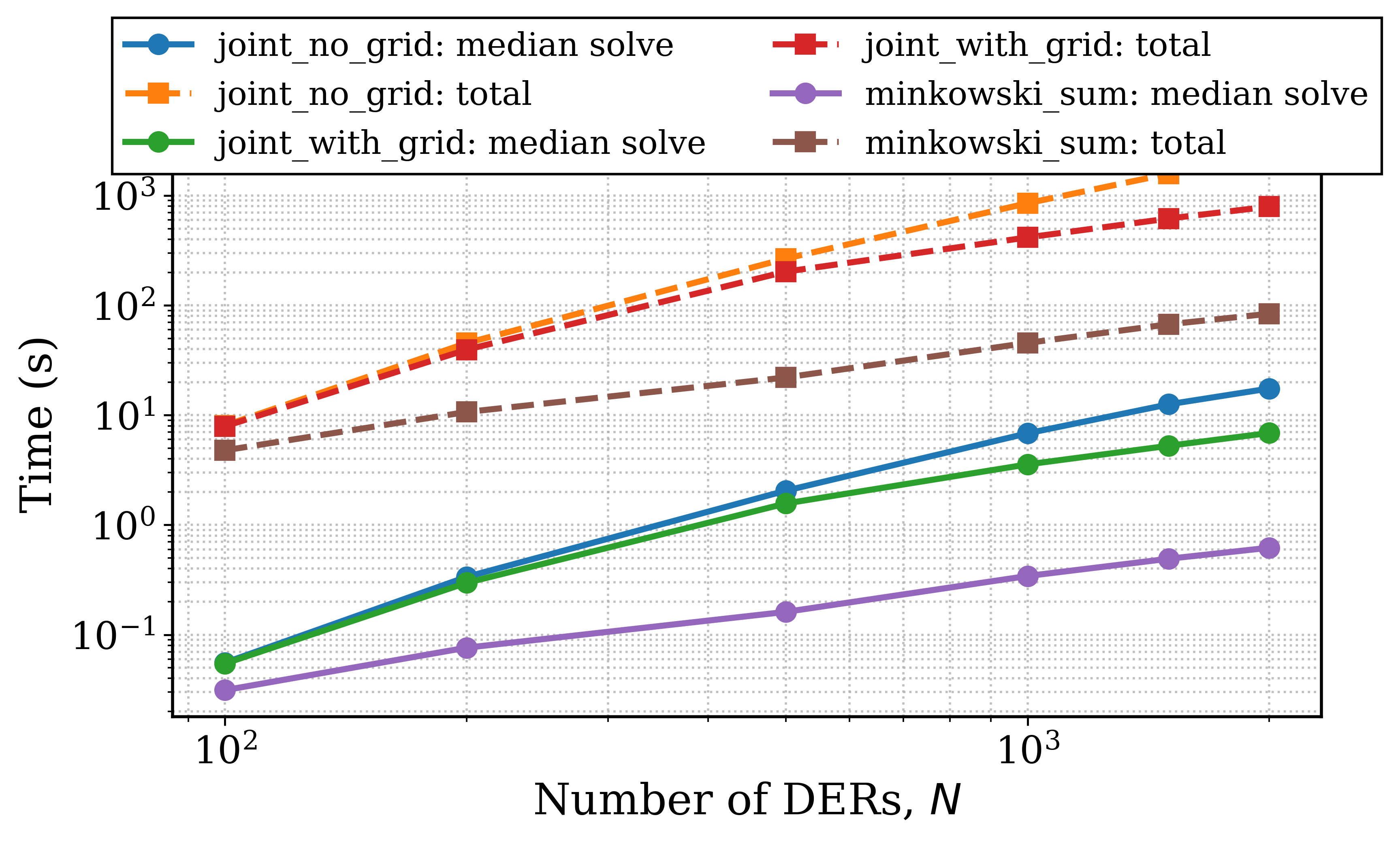}
  \vspace{-0.5em}
  \caption{Runtime versus the number of DERs $N$($M=128$)}
  \label{fig:runtime_vsN_M128}
\end{figure}


\begin{figure}[t]
  \centering
\includegraphics[width=0.82\columnwidth]{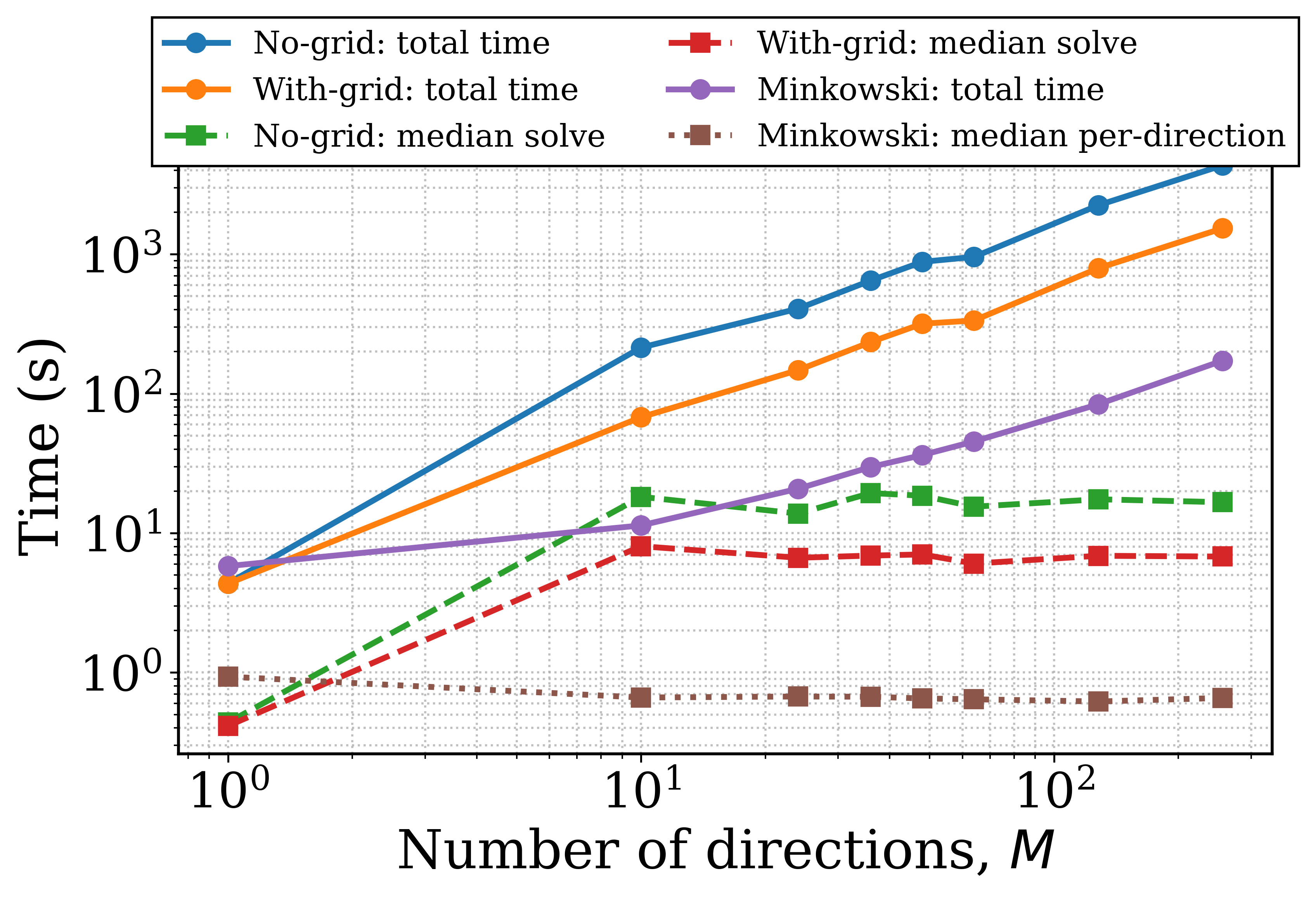}
\vspace{-0.5em}
\caption{Runtime versus the number of projection directions $M$($N=2000$)}
\label{fig:runtime_vsM_N2000}
\end{figure}


\section{Conclusion}
This paper presented a computationally transparent framework for
aggregating the power flexibility of heterogeneous DERs.
The framework represents device-level operating characteristics,
including conventional batteries, EVs, ZABs, SFBs, and PV systems,
through a unified linear formulation while accommodating feeder- and
portfolio-level coupling constraints. An adaptive directional support
method was developed to construct feasible inner approximations without
explicitly enumerating the facets or vertices of 
flexibility region.

Theoretical analysis established the feasibility and
non-overestimation properties of the resulting approximation and
characterized its computational complexity. In the small-scale case,
the proposed method provided substantially tighter projected regions
than the box, zonotope, and inner-ball benchmarks and matched the dense
directional reference within numerical tolerance using 16--64
construction directions. Large-scale simulations with up to 2000 DERs
showed approximately linear empirical scaling with the portfolio size
and direction budget over the tested ranges. These results demonstrate
that the proposed framework provides a scalable and theoretically
grounded approach to flexibility characterization. Future work will
consider uncertainty-aware models, adaptive high-dimensional direction
selection, parallel computation, and nonlinear distribution-network
constraints.

\bibliographystyle{IEEEtran}
\bibliography{references}

\newcommand{\MY}[1]{{\color{blue}[Zhai: {#1}]}}
  \newcommand{\Du}[1]{{\color{red}[Du: {#1}]}}

\end{document}